\documentclass[runningheads]{llncs}
\usepackage[T1]{fontenc}
\usepackage{microtype}
\usepackage[vlined,ruled,algo2e,linesnumbered]{algorithm2e}
\usepackage{tikz}
\usepackage{amsmath}
\usepackage{amssymb}
\usepackage{mathtools}
\usepackage{csquotes}

\usetikzlibrary{decorations.pathreplacing}

\usepackage{amsthm}

\usepackage{graphicx}
\newcommand{\OPT}{\mathrm{OPT}}

\newcommand{\calA}{\mathcal{A}}
\newcommand{\ALG}{\mathrm{ALG}}
\newcommand{\bbN}{\mathbb{N}}
\newtheorem*{theorem*}{Theorem}
\newtheorem*{quest*}{Question}
\begin{document}
\title{Almost Tight Bounds for\\ Online Hypergraph Matching}
%
%
\author{Thorben Tr\"obst\inst{1} \and Rajan Udwani\inst{2}}
\authorrunning{T. Tr\"obst and R. Udwani}
\institute{Department of Computer Science,\\
University of California, Irvine,\\
\email{t.troebst@uci.edu}\\
\and\
Department of Industrial Engineering and Operations Research,\\
University of California, Berkeley,\\
\email{rudwani@berkeley.edu}
}
\maketitle              
\begin{abstract}
    In the online hypergraph matching problem, 
    hyperedges of size $k$ over a common ground set arrive online in adversarial order.
    The goal is to obtain a maximum matching (disjoint set of hyperedges). A na\"ive greedy algorithm for this problem achieves a competitive ratio of $\frac{1}{k}$. We show that no (randomized) online algorithm has competitive ratio better than
    $\frac{2+o(1)}{k}$.
    If edges are allowed to be assigned fractionally, we give a deterministic online algorithm with competitive ratio
    $\frac{1-o(1)}{\ln(k)}$ and show that no online algorithm can have competitive ratio strictly better than $\frac{1+o(1)}{\ln(k)}$. Lastly, we give a $\frac{1-o(1)}{\ln(k)}$ competitive algorithm for the fractional
    \emph{edge-weighted} version of the problem under a free disposal assumption.

\end{abstract}
\section{Introduction}
In the classic problem of online bipartite matching (OBM), we have bipartite graph between resources (offline vertices) and agents (online vertices). Agents arrive sequentially and edges incident on an agent are revealed on arrival. Each agent must be matched to at most one available resource. Matches are irrevocable and the goal is to maximize the size of the overall matching.  We evaluate the performance of an online algorithm using its competitive ratio (CR). CR is the worst case ratio between the cardinality of online matching and the optimal offline solution, which, for OBM, is the maximum cardinality matching. 

 OBM captures the essence of resource allocation under sequential and heterogeneous demand \cite{karp1990optimal}. Fundamental to OBM and its many variations \cite{mehta2013online}, is the assumption that every agent seeks at most one type of resource. 
In a variety of settings, ranging from revenue management in airlines \cite{simpson1989using,talluri1998analysis} to combinatorial auctions \cite{kesselheim2013optimal,korula2009algorithms} and ridesharing \cite{pavone2022online}, agents require a bundle (set) of resources and an allocation that does not include every resource in the bundle has no value. Inspired by this, we study a fundamental generalization of OBM, namely, online hypergraph matching. A hypergraph is a generalization of a graph where an edge can join any number
of vertices. The maximum cardinality of an edge, or hyperedge, is the rank of the hypergraph, i.e. hyperedges of a rank $k$ hypergraph have at most $k$ vertices. A matching in a hypergraph is a set of disjoint edges.
In online hypergraph matching, edges are revealed sequentially in a (possibly) adversarial order. On arrival of an edge, we must make an immediate an irrevocable decision to include the edge in the matching or reject it forever. The objective is to maximize the size of the matching. In the fractional version of the problem, arriving edges can be fractionally included in the matching, subject to the constraint that the total fraction of edges incident on any vertex is at most one.  

In case of OBM, the offline optimum can be found in polynomial time. In sharp contrast, for a rank $k$ hypergraph it is NP hard to find a $\frac{\Omega(\log k)}{ k}$ approximation for the offline hypergraph matching problem, unless P=NP \cite{hazan2006complexity}. As $k$ increases, the separation between these problems widens since the hypergraph matching problem becomes harder to approximate.
In this paper, our goal is to find tight upper and lower bounds for online hypergraph matching for large $k$. We consider both the integral and fractional 
version of the online hypergraph matching problem.

In the integral case, it can be shown that a simple greedy algorithm that always includes an arriving edge in the matching, if feasible, is $\frac{1}{k}$ competitive. In fact, no deterministic algorithm can do better in the worst case. From the hardness of the offline problem, no polynomial time (online) algorithm can have competitive ratio better than $\frac{\Omega(\log k)}{ k}$ (unless P=NP). 
For the fractional case, where the computational hardness disappears\footnote{In the fraction version, we evaluate competitive ratio against a fractional relaxation of the offline problem that can solved in polynomial time by using linear programming.}, a result for the online packing problem \cite{buchbinder2009online}  gives a $O(\log k)$ competitive algorithm for online hypergraph matching. 
This exponential gap between the fractional and integral setting raises the following natural question:
\medskip

\noindent \emph{For some $\epsilon>0$, is there an (exponential time) online algorithm with competitive ratio $\frac{O(1)}{k^{1-\epsilon}}$ for (integral) online hypergraph matching?}
\medskip

We answer this question in the negative and establish the following result. Let $o(1)=f(k)$ for some non-negative function $f(k)$ such that $\lim_{k\to+\infty} f(k) =0$. 

\begin{theorem*}[Informal]
	No (randomized and exponential time) online algorithm can achieve a competitive ratio better than $\frac{2+o(1)}{k}$ for online hypergraph matching.
	\end{theorem*}

Unlike the state-of-the-art complexity theoretic upper bound of $\frac{\Omega(\log k)}{ k}$, our result is unconditional and arises from the online nature of the problem (as opposed to any computational barriers). For the fractional case, we give new upper and lower bounds, closing the existing constant factor gap for large $k$.
\begin{theorem*}[Informal]
	There is an efficient algorithm for fractional online hypergraph matching with competitive ratio $\frac{1-o(1)}{\ln k}$. This is the best possible (asymptotic) CR for any online algorithm.
	\end{theorem*}
\section{Model}\label{sec:model}
Consider a hypergraph $G$ with offline vertices (resources) $I$, online vertices  $t\in T$. A hyperedge $e\in E$ is a tuple $(S,t)$ where $S\subseteq I$.  
For a degree $k$ hypergraph, we have $|S|\leq k$ for every edge, i.e., an edge intersects at most $k$ offline vertices. In the vertex arrival model of online hypergraph matching, vertices $t\in T$ arrive sequentially. On arrival of a vertex $t$, all hyperedges incident on $t$ are revealed and we can irrevocably choose at most one edge, provided this edge does not intersect with any edge chosen prior to the arrival of $t$. The goal is to choose the maximum number of disjoint edges, i.e., maximize the size of the matching. In the  edge-arrival model, that we discussed previously and use more frequently in the paper, each online vertex has exactly one edge incident on it. In both models, the arrival sequence is adversarial.

In terms of competitive ratio guarantee, the vertex arrival model is equivalent to the edge arrival model for large $k$ (the regime of interest in this paper). To see this, observe that an $\alpha$ competitive algorithm for the vertex arrival model is also $\alpha$ competitive for the edge arrival model. Further, given an $\alpha$ competitive algorithm for degree $k+1$ hypergraphs in the edge arrival model, we can construct an online algorithm with CR $\alpha$ for degree $k$ hypergraphs in the vertex arrival model. 
The algorithm for vertex arrivals will emulate the edge arrival based algorithm on a transformed instance with edge arrivals. On arrival of a vertex $t\in T$ with incident edges $\{e_1,e_2,\cdots, e_m\}$ in the vertex arrival model, we create $m$ new arrivals and add an additional offline vertex $i_t$ in the edge arrival instance. For $\ell\leq m$, the $\ell$-th new arrival in the edge arrival instance has a single edge $E_\ell\cup\{i_t\}$. We input the $m$ edges one-by-one (in arbitrary order) to the algorithm for edge arrivals. If the algorithm chooses edge edge $\ell\in[m]$, then we choose edge $e_\ell$ for arrival $t$ in the vertex arrival model. This procedure gives an online algorithm for the vertex arrival model with CR $\alpha$. 

In the fractional version, the optimal offline solution in edge arrival model is a solution to the following linear program,
\begin{eqnarray*}
\textbf{Fractional LP:} \qquad \max_{\{x_e\}_{e\in E}} &\sum_{e\in E}x_e,\\
s.t.\, &\sum_{i\in e} x_e &\leq 1,\qquad \forall i\in I,\\
&x_e&\geq 0 \qquad \forall e\in E.
\end{eqnarray*}
Decision variable $x_e\in[0,1]$ in the LP captures the fraction of edge $e$ included in the matching. The constraint, $\sum_{i\in e} x_e \leq 1$, enforces total fraction of edges incident on a resource $i\in I$ to be at most 1. An online algorithm can include arbitrary fraction of each arriving edge, subject to the same constraint as the LP.

Finally, note that by adding dummy resources we can assume w.l.o.g., that every edge has exactly $k$ offline vertices. Unless stated otherwise, in the rest of the paper we consider the edge arrival model and assume that instances are $k$-uniform, i.e., every edge intersects $k$ offline vertices.

\subsection{Related Work}
Perhaps closest to our setting is the work of Buchbinder and Naor \cite{buchbinder2009online} on online packing. They considered an online packing problem that generalizes the fractional online hypergraph matching problem and gave a $O(\log k)$ competitive algorithm. Buchbinder and Naor \cite{buchbinder2009online} also showed that no online algorithm can have competitive ratio better than $\Omega(\log k)$ for fractional online hypergraph matching. A special case of the online packing problem was considered earlier in \cite{awerbuch1993throughput}, in the context of online routing. Another closely related line of work is on the problem of network revenue management (NRM) \cite{simpson1989using,talluri1998analysis}. This is a stochastic arrival setting where seats in flights are offline resources allocated to sequentially arriving customers. A customer with multi-stop itinerary requires a seat on each flight in the itinerary. 
Recently, Ma et al.\ \cite{ma2020approximation} gave a $\frac{1}{k+1}$ algorithm for NRM. 
Another stream of work has focused on hypergraph matching from the perspective of ridesharing. Pavone et al.\ \cite{pavone2022online} introduced a hypergraph matching problem with deadlines to capture applications in ridesharing. Their model and results are incomparable to ours. Lowalekar et al.\ \cite{lowalekar2020competitive} consider a model inspired by ridesharing but with a stochastic arrival sequence. Finally,  \cite{kesselheim2013optimal,korula2009algorithms} consider related settings in combinatorial auctions that correspond to online hypergraph matching with stochastic arrivals.

For the offline hypergraph matching problem, Hazan et al.\ \cite{hazan2006complexity} showed that unless P=NP, no polynomial time algorithm can find an matching better than $O(\frac{\log k}{k})$ of the optimum matching. \cite{chan2012linear,cygan2013improved} give approximation algorithms for the problem. To the best of our knowledge, the state-of-the-art result is a $\frac{3}{k}$ approximation due to Cygan \cite{cygan2013improved}. 

\section{Integral Matchings}

In the following, fix some $k \geq 2$.
We will focus on the online hypergraph matching problem in the edge arrival model.
We start with a result that is folklore in the literature on online matching \cite{mehta2013online}.
\begin{theorem}[Folklore]
    There is a $\frac{1}{k}$-competitive algorithm for the
    online hypergraph matching problem. This is the best possible CR for deterministic algorithms.
\end{theorem}

\begin{proof}
    Consider the online algorithm that includes an arriving edge $e$ in the matching if it is disjoint with all previously included edges. 
    Let $e_1, \ldots, e_\ell$ be the set of edges included in an offline optimum solution and let $I_o\subseteq I$ denote the set of offline vertices that are covered by the edges chosen in the online algorithm.
    For each edge $e_i$, at least one of the resources that it intersects 
    must be included in $I_o$. Thus, $|I_o|\geq \frac{1}{k} \ell\,k$ and
    the online algorithm picks at least $\frac{1}{k}\ell$ hyper edges. 
    
    To see that this is the best possible CR, consider two arrival sequences. In the first sequence, we have a single arrival. In the second sequence, we augment the first sequence with $k$ more arrivals such that the optimal offline matching has size $k$ but the first edge intersects every other edge. A deterministic algorithm with non-zero competitive ratio must match the first arrival.
\end{proof}
In our first result, we show that even a randomized and possibly exponential time online algorithm cannot achieve a much better competitive ratio for this problem. To show this, we will use Yao's minimax principle, as stated below.
\begin{lemma}[Yao's Principle]\label{lem:yao}
    Let $\alpha$ be the best competitive ratio of any randomized algorithm.
    Let $\beta$ be the competitive ratio of the best deterministic algorithm against some fixed
    distribution of instances.
    Then $\alpha \leq \beta$.
\end{lemma}

Before we get to our main result, 
we will first give a slightly weaker result
that serves both as a warm up and as a gadget for the main result. Recall that, 
$o(1)=f(k)$ for some non-negative function $f(k)$ such that $\lim_{k\to+\infty} f(k) =0$.  

\begin{theorem}\label{thm:4_over_k}
    There does not exist a $\frac{4+o(1)}{k}$ competitive algorithm
    for the $k$-uniform online hypergraph matching problem for any $\epsilon > 0$.
\end{theorem}

\begin{proof}
    Using Yao's principle, we will construct a distribution of instances with even $k$ where $\OPT = \frac{k}{2}$ but
    the best deterministic online algorithm can only achieve an expected matching size of $2$.

    For any given even value $k$, the overall (random) instance $G_k$ will consist of $\frac{k}{2}$ \enquote{red} edges and $\frac{k}{2}$ \enquote{blue} edges
    constructed in $\frac{k}{2}$ phases.
    See Figure~\ref{fig:4_over_k_ex} for an example.
    In each phase, do the following:

    \begin{enumerate}
        \item Let $A$ be a set of vertices such that every vertex in $A$ is contained in exactly
            one blue edge and every blue edge contains exactly one vertex from $A$.
            Create a new edge $e_1$ which consists of $A$ and $k - |A|$ many new vertices that
            have not been in any edges yet.
            Let $e_1$ arrive in the instance.
        \item Repeat step~1 to get another edge $e_2$ with the same properties (note that, in
            paritcular, $e_2$ will intersect $e_1$).
            Let $e_2$ arrive in the instance.
        \item Randomly let one of $\{e_1, e_2\}$ be red and the other blue with equal probability.
    \end{enumerate}

    The crucial properties of this construction are that each blue edge intersects all future edges
    whereas each red edge is disjoint from all future edges.
    In particular, the $\frac{k}{2}$ red edges form a maximum size matching, i.e.\ $\OPT =
    \frac{k}{2}$.

    Now consider some deterministic online algorithm $\calA$.
    Let $\alpha_i$ be the probability that $\calA$ matches the red edge in phase $i$ and let
    $\beta_i$ be the probability that $\calA$ matches the blue edge in phase $i$.
    Clearly, since the red and blue edges are determined independently and uniformly at random, we
    must have $\alpha_i = \beta_i$.
    Moreover, since at most one blue edge can be picked, we know $\alpha_1 + \dots + \alpha_{k /
    2} \leq 1$.
    Thus the expected size of the matching generated by $\calA$ is at most
    \[
        \alpha_1 + \dots + \alpha_{k / 2} + \beta_1 + \dots + \beta_{k / 2} \leq 2. \qedhere
    \]
\end{proof}

\begin{figure}[htb]
    \centering
    \begin{tikzpicture}

        \draw[rounded corners, red] (0.1, 0.6) -- (4.9, 0.6) -- (4.9, 0.9) -- (0.1, 0.9) -- cycle;
        \draw[rounded corners, red] (0.1, 1.1) -- (4.9, 1.1) -- (4.9, 1.4) -- (0.1, 1.4) -- cycle;
        \draw[rounded corners, red] (0.1, 1.6) -- (4.9, 1.6) -- (4.9, 1.9) -- (0.1, 1.9) -- cycle;
        \draw[rounded corners, red] (0.1, 2.1) -- (4.9, 2.1) -- (4.9, 2.4) -- (0.1, 2.4) -- cycle;
        \draw[rounded corners, red] (0.1, 2.6) -- (4.9, 2.6) -- (4.9, 2.9) -- (0.1, 2.9) -- cycle;

        \draw[rounded corners, blue] (0.1, 0.6) -- (0.4, 0.6) -- (0.4, 5.4) -- (0.1, 5.4) -- cycle;
        \draw[rounded corners, blue] (0.6, 1.1) -- (0.9, 1.1) -- (0.9, 5.4) -- (0.6, 5.4) -- (0.6, 4.9)
        -- (0.1, 4.9) -- (0.1, 4.6) -- (0.6, 4.6) -- cycle;
        \draw[rounded corners, blue] (1.1, 1.6) -- (1.4, 1.6) -- (1.4, 5.4) -- (1.1, 5.4) -- (1.1, 4.4)
        -- (0.1, 4.4) -- (0.1, 4.1) -- (1.1, 4.1) -- cycle;
        \draw[rounded corners, blue] (1.6, 2.1) -- (1.9, 2.1) -- (1.9, 5.4) -- (1.6, 5.4) -- (1.6, 3.9)
        -- (0.1, 3.9) -- (0.1, 3.6) -- (1.6, 3.6) -- cycle;
        \draw[rounded corners, blue] (2.1, 2.6) -- (2.4, 2.6) -- (2.4, 5.4) -- (2.1, 5.4) -- (2.1, 3.4)
        -- (0.1, 3.4) -- (0.1, 3.1) -- (2.1, 3.1) -- cycle;

        \node[anchor=west, red] at (5, 0.75) {\small $1$};
        \node[anchor=west, red] at (5, 1.25) {\small $2$};
        \node[anchor=west, red] at (5, 1.75) {\small $3$};
        \node[anchor=west, red] at (5, 2.25) {\small $4$};
        \node[anchor=west, red] at (5, 2.75) {\small $5$};

        \node[anchor=south, blue] at (0.25, 5.5) {\small $1$};
        \node[anchor=south, blue] at (0.75, 5.5) {\small $2$};
        \node[anchor=south, blue] at (1.25, 5.5) {\small $3$};
        \node[anchor=south, blue] at (1.75, 5.5) {\small $4$};
        \node[anchor=south, blue] at (2.25, 5.5) {\small $5$};

        \foreach \y in {0.75, 1.25, 1.75, 2.25, 2.75} {
            \foreach \x in {0.25, 0.75, 1.25, 1.75, 2.25, 2.75, 3.25, 3.75, 4.25, 4.75} {
                \node[circle, fill, inner sep=0.5] at (\x, \y) {};
            }
        }

        \foreach \y in {3.25, 3.75, 4.25, 4.75, 5.25} {
            \foreach \x in {0.25, 0.75, 1.25, 1.75, 2.25} {
                \node[circle, fill, inner sep=0.5] at (\x, \y) {};
            }
        }

        \fill[rounded corners, fill opacity=0.1] (0.1, 3.1) -- (2.4, 3.1) -- (2.4, 3.4) -- (0.1, 3.4) -- cycle;
        \fill[rounded corners, fill opacity=0.1] (0.1, 3.6) -- (2.4, 3.6) -- (2.4, 3.9) -- (0.1, 3.9) -- cycle;
        \fill[rounded corners, fill opacity=0.1] (0.1, 4.1) -- (2.4, 4.1) -- (2.4, 4.4) -- (0.1, 4.4) -- cycle;
        \fill[rounded corners, fill opacity=0.1] (0.1, 4.6) -- (2.4, 4.6) -- (2.4, 4.9) -- (0.1, 4.9) -- cycle;
        \fill[rounded corners, fill opacity=0.1] (0.1, 5.1) -- (2.4, 5.1) -- (2.4, 5.4) -- (0.1, 5.4) -- cycle;

        \node[anchor=west] at (2.5, 3.25) {\small $A_1$};
        \node[anchor=west] at (2.5, 3.75) {\small $A_2$};
        \node[anchor=west] at (2.5, 4.25) {\small $A_3$};
        \node[anchor=west] at (2.5, 4.75) {\small $A_4$};
        \node[anchor=west] at (2.5, 5.25) {\small $A_5$};
    \end{tikzpicture}
    \caption{Shown is gadget $G_{10}$ proving that a competitive ratio of $\frac{4}{k} + \epsilon$ is
    imposisble for $k = 10$. The numbers indicate in which phase each edge was added. The lightly
    shaded areas represent the vertex sets $A_1, \ldots, A_5$ which are useful for the construction
    of $H_k$.}
    \label{fig:4_over_k_ex}
\end{figure}
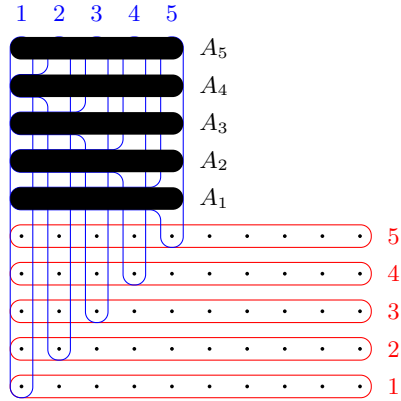

\begin{theorem}
    There does not exist a $\frac{2+o(1)}{k}$ competitive
    algorithm for the online hypergraph matching problem for any $\epsilon > 0$.
\end{theorem}

\begin{proof}[Proof Sketch] Due to space limitation, we give an outline of the main steps.
    We will use induction to create a distribution over graphs $H_k$ for all $k$, with the following properties:
    \begin{enumerate}
        \item There are $k$ red and $k$ blue edges.
        \item The edges appear in $k$ phases, each of which consists of one red and one blue edge
            where the color is chosen uniformly and independently at random.
        \item Every blue edge intersects all future edges.
        \item Every red edge is disjoint from all future edges.
    \end{enumerate}

    $H_1$ is trivial to construct.
    We will just have a single vertex which is simultaneously in both a red and blue singleton edge. Suppose that we can construct $H_\ell$, for all $1\leq\ell\leq k-1$.

    Now in order to construct $H_k$, we first employ the $\frac{k}{2}$ phases of $G_k$.
    After this we can construct $\frac{k}{2}$ disjoint sets $A_1, \ldots, A_{k / 2}$ of
    $\frac{k}{2}$ vertices
    each such that each $A_i$ intersects all blue edges and none of the red edges.
    See again Figure~\ref{fig:4_over_k_ex}. Now, for the remaining $k / 2$ phases, we recursively employ the distribution $H_{k / 2}$ as follows. In phase $i+\frac{k}{2}$, we extend the the two edges from phase $i$ of
    $H_{k / 2}$ by the set $A_i$ to form sets of size $k$. This gives us the two edges of rank $k$ for phase $i+\frac{k}{2}$. 

    Finally, one may check that properties $1-4$ are satisfied by induction.
    Thus, we may conclude the proof similar to the proof of Theorem~\ref{thm:4_over_k}: the optimum
    solution picks all $k$ red edges whereas any deterministic online algorithm can only get an
    expected value of $2$ since at most one blue edge can be picked and red and blue edges in each phase are indistinguishable for the online algorithm.
\end{proof}

\section{Fractional Matchings}

Inspired by the Waterfilling (or Balancing) algorithm of \cite{kalyanasundaram2000optimal,mehta2007adwords}, designed for online bipartite matching and its variants, we propose the following online algorithm for fractional hypergraph matching in the edge arrival model.

\medskip
\begin{algorithm2e}[H]
    For each $e \in E$, set $y_e \coloneqq 0$.\\
    \For{each edge $e$ which arrives}{
        Increase $y_e$ continuously as long as $\sum_{i \in e} (k \, \ln(k))^{x_i - 1}
        \leq 1$ where $x_i \coloneqq \sum_{f \in E: i \in f} y_f$.
    }
    \caption{\textsc{Hypergraph Water-Filling}\label{alg:water_filling}}
\end{algorithm2e}
\medskip
The final value of variable $y_e$ in the algorithm is the fraction of edge $e$ that is included in the matching. At any moment, variable $x_i$ in the algorithm captures the total fraction of edges incident on resource $i$ that have been included in the matching. In other words, the value of $x_i$ is the fraction of $i$ that has already been ``matched" by the algorithm. In order to ``preserve" resources for future edges, Algorithm \ref{alg:water_filling} stops matching an edge when the value of $\sum_{i \in e} (k \, \ln(k))^{x_i - 1}$ reaches 1. For illustration, consider the following scenarios at the arrival 
of edge $e$ incident on resources $\{1,2,\cdots,k\}$. 
\smallskip

\noindent  $(i)$ $x_i=0\,\, \forall i\in[k]$ on arrival of $e$. Then, the algorithm will match $\frac{\ln (\ln k)}{\ln (k)+\ln (\ln k)}$ fraction of the edge before stopping. 
\smallskip

\noindent $(ii)$ $x_i=0.5\,\, \forall i\leq \sqrt{k \ln (k)}$ and $x_i=0$ otherwise. Then, for sufficiently large values of $k$, the algorithm does not match any fraction of edge $e$.

We would like to note that Algorithm \ref{alg:water_filling} constructs a fractional matching by augmenting the primal
solution $y$, whereas the online packing algorithm of \cite{buchbinder2009online} augments the dual solution. We show that Algorithm 1 achieves the best possible competitive ratio guarantee for large $k$. 
\begin{theorem}\label{thm:frac_lb}
    For some non-negative function $f(k)$ such that $\lim_{k\to+\infty} f(k)=0$,  Algorithm~\ref{alg:water_filling} is $\frac{1 -f(k)}{\ln(k)}$-competitive for the fractional online hypergraph matching problem.
\end{theorem}

\begin{proof}
   Given a hypergraph $G$ with resources $I$ and set of hyperedges $E$, let $\ALG$ denote the total
    revenue of the online algorithm and let $\OPT$ denote the value of the optimal fractional offline solution. We use a primal-dual approach inspired by \cite{buchbinder2009online,devanur2013randomized} to prove the result. Using weak duality, it suffices to find a feasible solution to the following system of linear inequalities,
    \begin{eqnarray}
    \sum_{e\in E} u_e + \sum_{i\in I} r_i &\leq &\ALG,\label{dual1}\\
    u_e + \sum_{i \in e} r_i &\geq &\frac{1 - \frac{1 }{\ln(k)}}{\ln(k) + \ln(\ln(k))} \quad \forall e\in E.\label{dual2}
    \end{eqnarray}
    We set the dual variables using the following procedure. In the beginning, all variables are set to 0. When Algorithm~\ref{alg:water_filling} is matching edge $e$ in line~3, we interpret the
    quantity $\sum_{i \in e} (k \ln(k))^{x_i - 1}$ as the \emph{price} of
    edge $e$.
    Accordingly, if $e$ is matched some by some infinitesimal amount $\mathrm{d} t$, then, 
 \begin{enumerate}
     \item For every resource $i\in e$, we increase the \emph{revenue} $r_i$ by
    $(k \ln(k))^{x_i - 1} \mathrm{d} t$. 
    \item  We increase the \emph{utility} $u_e$ of $e$ by
    $\left(1 - \sum_{i \in e} (k \ln(k))^{x_i - 1}\right) \mathrm{d} t$.
 \end{enumerate} 
    Note that this implies that the total sum of all revenues and utilities is equal to the total
    size of the matching, i.e., inequality \eqref{dual1} is satisfied. 





 It remains to show inequality \eqref{dual2} for every edge. Fix an arbitrary edge $e\in E$ and consider the following two cases.
 
   \noindent \textbf{Case 1:}  Let $e \in E$ be arbitrary and let $x_i$ be the final fill levels of the vertices $i \in e$.
    If $x_i = 1$ for any $i$, we know that
    \[
        r_i = \int_0^1{(k \ln(k))^{t - 1} \, \mathrm{d} t} = \frac{1 - 1 / (k \ln(k))}{\ln(k)
        + \ln(\ln(k))} \geq \frac{1 - 1  / \ln(k)}{\ln(k) + \ln(\ln(k))},
    \]
    so the revenue of this one vertex is already enough for the claim.

    \noindent \textbf{Case 2:} Otherwise, we can assume that $x_i < 1$ for all $i$.
    But in that case, at any time at which $e$ was being matched by some $\mathrm{d} t$, $e$ would
    have gotten a utility of at least,
    \[
        \max\left\{0,\left(1 - \sum_{i \in e} (k \ln(k))^{x_i - 1}\right) \mathrm{d} t\right\}.
    \]
    Let us denote $P \coloneqq \sum_{i \in e} (k \ln(k))^{x_i - 1}$, then this
    implies that $u_e \geq \max \{0, 1 - P\}$.
    In total, we get
    \begin{align*}
        u_e + \sum_{i \in e} r_i &\geq \max \{0, 1 - P\} + \sum_{i \in e} \int_0^{x_i}
        (k \ln(k))^{t - 1} \, \mathrm{d} t \\
        &= \max \{0, 1 - P\} + \frac{P - 1 / \ln(k)}{\ln(k) + \ln(\ln(k))}.
    \end{align*}

    Finally, observe that this quantity is minimized if $P = 1$ under the assumption that $\ln(k)
    + \ln(\ln(k)) \geq 1$.
    Thus the claim, and the theorem, follows.
\end{proof}

Next, we show that no online algorithm can beat the performance of Algorithm 1.

\begin{theorem}\label{thm:frac_ub}
    For some non-negative function $f(k)$ such that $\lim_{k\to+\infty} f(k)=0$,  there does not exist any online algorithm which is $\frac{1 +
    f(k)}{\ln(k)}$-competitive for online hypergraph matching
    problem. 
\end{theorem}

\begin{figure}[h]
    \centering
    \begin{tikzpicture}
        \foreach \x in {0.25, 0.75, 1.25} {
            \foreach \y in {0.25, 0.75, 1.25, 1.75, 2.25, 2.75, 3.25, 3.75, 4.25, 4.75} {
                \node[circle, fill, inner sep=0.5] at (\x, \y) {};
            }
        }

        \draw[rounded corners, blue, fill, fill opacity=0.05] (0.1, 0.1) -- (1.4, 0.1) -- (1.4, 1.9) --
        (1.1, 1.9) -- (1.1, 1.4) -- (0.1, 1.4) -- cycle;
        \draw[rounded corners, blue, fill, fill opacity=0.05] (0.1, 1.6) -- (0.9, 1.6) -- (0.9, 2.1) --
        (1.4, 2.1) -- (1.4, 3.4) -- (0.6, 3.4) -- (0.6, 2.9) -- (0.1, 2.9) -- cycle;
        \draw[rounded corners, blue, fill, fill opacity=0.05] (0.1, 3.1) -- (0.4, 3.1) -- (0.4, 3.6) --
        (1.4, 3.6) -- (1.4, 4.9) -- (0.1, 4.9) -- cycle;

        \draw[rounded corners, red, fill, fill opacity=0.1] (0.1, 0.1) -- (1.4, 0.1) -- (1.4, 0.9)
        -- (0.1, 0.9) -- cycle;
        \draw[rounded corners, red, fill, fill opacity=0.1] (0.1, 1.1) -- (1.4, 1.1) -- (1.4, 1.9)
        -- (0.1, 1.9) -- cycle;
        \draw[rounded corners, red, fill, fill opacity=0.1] (0.1, 2.1) -- (1.4, 2.1) -- (1.4, 2.9)
        -- (0.1, 2.9) -- cycle;
        \draw[rounded corners, red, fill, fill opacity=0.1] (0.1, 3.1) -- (1.4, 3.1) -- (1.4, 3.9)
        -- (0.1, 3.9) -- cycle;
        \draw[rounded corners, red, fill, fill opacity=0.1] (0.1, 4.1) -- (1.4, 4.1) -- (1.4, 4.9)
        -- (0.1, 4.9) -- cycle;

        \begin{scope}[shift={(0.5, 0)}]
            \foreach \x in {2.25, 2.75} {
                \foreach \y in {0.25, 0.75, 1.25, 1.75, 2.25, 2.75, 3.25, 3.75, 4.25} {
                    \node[circle, fill, inner sep=0.5] at (\x, \y) {};
                }
            }

            \draw[rounded corners, blue, fill opacity=0.05] (2.1, 0.1) -- (2.9, 0.1) -- (2.9, 1.4)
            -- (2.1, 1.4) -- cycle;
            \draw[rounded corners, blue, fill opacity=0.05] (2.1, 1.6) -- (2.9, 1.6) -- (2.9, 2.9)
            -- (2.1, 2.9) -- cycle;
            \draw[rounded corners, blue, fill opacity=0.05] (2.1, 3.1) -- (2.9, 3.1) -- (2.9, 4.4)
            -- (2.1, 4.4) -- cycle;

            \draw[rounded corners, red, fill, fill opacity=0.1] (2.1, 0.1) -- (2.9, 0.1) -- (2.9,
            0.9) -- (2.1, 0.9) -- cycle;
            \draw[rounded corners, red, fill, fill opacity=0.1] (2.1, 1.1) -- (2.9, 1.1) -- (2.9,
            1.9) -- (2.1, 1.9) -- cycle;
            \draw[rounded corners, red, fill, fill opacity=0.1] (2.1, 2.1) -- (2.9, 2.1) -- (2.9,
            2.9) -- (2.1, 2.9) -- cycle;
            \draw[rounded corners, red, fill, fill opacity=0.1] (2.1, 3.1) -- (2.9, 3.1) -- (2.9,
            3.9) -- (2.1, 3.9) -- cycle;
        \end{scope}

        \begin{scope}[shift={(1, 0)}]
            \foreach \x in {3.75, 4.25} {
                \foreach \y in {0.25, 0.75, 1.25, 1.75, 2.25, 2.75} {
                    \node[circle, fill, inner sep=0.5] at (\x, \y) {};
                }
            }

            \draw[rounded corners, blue, fill opacity=0.05] (3.6, 0.1) -- (4.4, 0.1) -- (4.4, 0.9)
            -- (3.6, 0.9) -- cycle;
            \draw[rounded corners, blue, fill opacity=0.05] (3.6, 1.1) -- (4.4, 1.1) -- (4.4, 1.9)
            -- (3.6, 1.9) -- cycle;
            \draw[rounded corners, blue, fill opacity=0.05] (3.6, 2.1) -- (4.4, 2.1) -- (4.4, 2.9)
            -- (3.6, 2.9) -- cycle;

            \draw[rounded corners, red, fill, fill opacity=0.1] (3.6, 0.1) -- (4.4, 0.1) -- (4.4,
            0.4) -- (3.6, 0.4) -- cycle;
            \draw[rounded corners, red, fill, fill opacity=0.1] (3.6, 0.6) -- (4.4, 0.6) -- (4.4,
            0.9) -- (3.6, 0.9) -- cycle;
            \draw[rounded corners, red, fill, fill opacity=0.1] (3.6, 1.1) -- (4.4, 1.1) -- (4.4,
            1.4) -- (3.6, 1.4) -- cycle;
            \draw[rounded corners, red, fill, fill opacity=0.1] (3.6, 1.6) -- (4.4, 1.6) -- (4.4,
            1.9) -- (3.6, 1.9) -- cycle;
            \draw[rounded corners, red, fill, fill opacity=0.1] (3.6, 2.1) -- (4.4, 2.1) -- (4.4,
            2.4) -- (3.6, 2.4) -- cycle;
            \draw[rounded corners, red, fill, fill opacity=0.1] (3.6, 2.6) -- (4.4, 2.6) -- (4.4,
            2.9) -- (3.6, 2.9) -- cycle;
        \end{scope}

        \begin{scope}[shift={(1.5, 0)}]
            \foreach \x in {5.25, 5.75} {
                \foreach \y in {0.25, 0.75, 1.25} {
                    \node[circle, fill, inner sep=0.5] at (\x, \y) {};
                }
            }

            \draw[rounded corners, blue, fill opacity=0.05] (5.1, 0.1) -- (5.9, 0.1) -- (5.9, 0.4)
            -- (5.1, 0.4) -- cycle;
            \draw[rounded corners, blue, fill opacity=0.05] (5.1, 0.6) -- (5.9, 0.6) -- (5.9, 0.9)
            -- (5.1, 0.9) -- cycle;
            \draw[rounded corners, blue, fill opacity=0.05] (5.1, 1.1) -- (5.9, 1.1) -- (5.9, 1.4)
            -- (5.1, 1.4) -- cycle;

            \draw[rounded corners, red, fill, fill opacity=0.1] (5.1, 0.1) -- (5.4, 0.1) -- (5.4,
            0.4) -- (5.1, 0.4) -- cycle;
            \draw[rounded corners, red, fill, fill opacity=0.1] (5.1, 0.6) -- (5.4, 0.6) -- (5.4,
            0.9) -- (5.1, 0.9) -- cycle;
            \draw[rounded corners, red, fill, fill opacity=0.1] (5.1, 1.1) -- (5.4, 1.1) -- (5.4,
            1.4) -- (5.1, 1.4) -- cycle;

            \draw[rounded corners, red, fill, fill opacity=0.1] (5.6, 0.1) -- (5.9, 0.1) -- (5.9,
            0.4) -- (5.6, 0.4) -- cycle;
            \draw[rounded corners, red, fill, fill opacity=0.1] (5.6, 0.6) -- (5.9, 0.6) -- (5.9,
            0.9) -- (5.6, 0.9) -- cycle;
            \draw[rounded corners, red, fill, fill opacity=0.1] (5.6, 1.1) -- (5.9, 1.1) -- (5.9,
            1.4) -- (5.6, 1.4) -- cycle;
        \end{scope}

        \begin{scope}[shift={(2, 0)}]
            \foreach \x in {6.75} {
                \foreach \y in {0.25, 0.75, 1.25} {
                    \node[circle, fill, inner sep=0.5] at (\x, \y) {};
                }
            }

            \draw[rounded corners, blue, fill opacity=0.05] (6.6, 0.1) -- (6.9, 0.1) -- (6.9, 0.4)
            -- (6.6, 0.4) -- cycle;
            \draw[rounded corners, blue, fill opacity=0.05] (6.6, 0.6) -- (6.9, 0.6) -- (6.9, 0.9)
            -- (6.6, 0.9) -- cycle;
            \draw[rounded corners, blue, fill opacity=0.05] (6.6, 1.1) -- (6.9, 1.1) -- (6.9, 1.4)
            -- (6.6, 1.4) -- cycle;
        \end{scope}

        \node at (2, 2.25) {$\Rightarrow$};
        \node at (4, 1.5) {$\Rightarrow$};
        \node at (6, 0.75) {$\Rightarrow$};
        \node at (8, 0.75) {$\Rightarrow$};
    \end{tikzpicture}
    \caption{Shown is the upper-bounding construction with $k = 10$, $l = 3$, $\delta = 0.5$.
    In each step we replace the blue edges with as many red edges of $\frac{1}{1 + \delta}$ times
    the size as possible.
    Then we pick the $l$ red edges that the algorithm puts the most weight on, make those the new
    blue edges and repeat until only singleton edges are left.}
    \label{fig:frac_ub}
\end{figure}
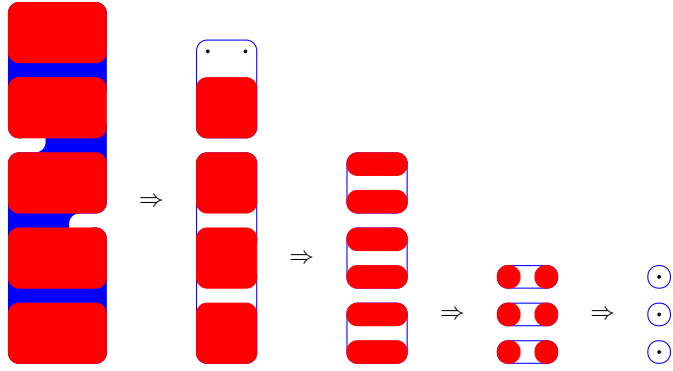

\begin{proof}[Proof Sketch] Due to space limitation, we give an outline of the main steps.
    Let $\calA$ be some algorithm for the fractional online hypergraph matching
    problem.
    We can assume w.l.o.g., that $\calA$ is deterministic.
    This is because if $\calA$ is randomized, we may create another algorithm $\calA'$ which simply
    fractionally allocates every edge $e$ with the expected value of $\calA$.
    Then $\calA'$ is a deterministic algorithm that performs just as well as $\calA$.

    Fix some large $l \in \bbN$ and small $\delta > 0$.
    We will now construct an instance in with some hyperedges may have less than $k$ vertices. 
    The instance is created according to the following procedure (see Figure~\ref{fig:frac_ub}).
    \begin{enumerate}
        \item Set $m \leftarrow k$ and let $l$ disjoint edges of size $k$ arrive.
            Let $U$ be the set of all vertices in these $l$ edges.
        \item If $m = 0$, stop. Otherwise, set $m \leftarrow \left\lfloor \frac{m}{1 + \delta}
            \right\rfloor$.
        \item Partition $U$ into as many disjoint edges of size $m$ as possible and let these
            arrive.
        \item Let $e_1, \ldots, e_l$ be the $l$ of these edges that $\calA$ matches the most.
        \item Update $U \leftarrow e_1 \cup \dots \cup e_l$ and go back to step~2.
    \end{enumerate}

    Now let $\alpha$ be the competitive ratio of $\calA$.
    Our first observation is that steps 2--5 execute $(1 - o(1)) \log_{1 + \delta}(k)$ many times as $k \rightarrow \infty$.
    Moreover, in each iteration, we cover $l$ edges with
    \[
        \left\lfloor \frac{l m}{\left\lfloor \frac{m}{1 + \delta}\right\rfloor} \right\rfloor \geq
        (1 + \delta) l - 1
    \]
    edges.
    Thus $\ALG \geq \alpha \OPT \geq \alpha (1 - o(1)) \log_{1 + \delta}(k) (\delta l - 1)$.

    Let $E^\star \subseteq E$ be the set of edges which are picked at various points in step~4 and
    let $y$ be the fractional matching constructed by $\calA$.
    Then because these edges are always the $l$ most covered edges we know that
    \[
        y(E^\star) \geq \min_{m \geq 1} \frac{l}{\left\lfloor \frac{l m}{\left\lfloor \frac{m}{1 +
        \delta} \right\rfloor} \right\rfloor} \ALG \geq \frac{1}{1 + \delta - \frac{1}{l}} \ALG.
    \]

    Lastly, we know that that all edges in $E^\star$ overlap in $l$ vertices, namely the $l$
    vertices that are contained in the final iteration of the loop.
    This implies that $y(E^\star) \leq l$.
    Combining these inequalities, we thus get
    \begin{align*}
        \alpha &\leq \frac{\left(1 + \delta - \frac{1}{l}\right) l}{(1 - o(1)) \log_{1 +
        \delta}(k) (\delta l - 1)} \\
        &= \frac{((1 + \delta) l - 1) \ln(1 + \delta) }{(1 - o(1)) (\delta l - 1)} \cdot
        \frac{1}{\ln(k)}
    \end{align*}

    Finally, observe that for small $\delta$, large $l$ and large $k$, we get $\alpha < \frac{1 +
    \epsilon}{\ln(k)}$ as claimed.
\end{proof}

\section{Edge-Weights}
 In this section, we consider the free disposal model for online weighted matching \cite{feldman2009online}, and generalize Algorithm 1 to online fractional weighted hypergraph matching with free disposal.
 In this model, we are allowed to drop previously matched edges with no penalty (though of course we
 do not count such dropped edges towards our objective function).

\medskip
\begin{algorithm2e}[H]
    For each $e \in E$, let $y_e \coloneqq 0$. \\
    For each $i \in V$, let $f_i(t) \coloneqq \sum_{e : i \in
    e, w_e \geq t} y_e$ for all $t \geq 0$. \\
    \For{each edge $e$ which arrives}{
        \While {$\sum_{i \in e} \int_0^{w_e} (k \ln(k))^{f_i(t) - 1} \, \mathrm{d} t
        \leq 1$} {
            \For{$i \in e$ with $x_i = 1$} {
                Let $e^-_i$ be a minimum weight edge with $i \in e^-_i$ and $y_{e^-_i}
                > 0$. \\
                $y_{e^-_i} \leftarrow y_{e^-_i} - \mathrm{d} x$
            }
            $y_e \leftarrow y_e + \mathrm{d} x$
        }
    }
    \caption{\textsc{Hypergraph Weighted Water-Filling}\label{alg:weighted_wf}}
\end{algorithm2e}
\medskip

\begin{theorem}\label{thm:edge_weights}
   For some non-negative function $f(k)$ such that $\lim_{k\to+\infty} f(k)=0$,
    Algorithm~\ref{alg:weighted_wf} is $\frac{1 -f(k)}{\ln(k)}$-competitive for online fractional weighted hypergraph
    matching problem with free disposal.
\end{theorem}

\begin{proof}[Proof Sketch] Due to space limitation, we give an outline of the main steps.
    The proof will use a similar primal-dual approach as Theorem~\ref{thm:frac_lb}.
    Once again, when Algorithm~\ref{alg:weighted_wf} is matching an edge $e$ by some amount
    $\mathrm{d} x$, we split the resulting gain into \emph{revenue} and \emph{utility} through the
    price $\sum_{i \in e} \int_0^{w_e} (k \ln(k))^{f_i(t) - 1} \, \mathrm{d} t$.
    However, this time the price is split into three parts:

    \begin{enumerate}
        \item Each $i \in e$, earns a revenue of $\sum_{i \in e} \int_{w_{e^-_i}}^{w_e}
            (k \ln(k))^{f_i(t) - 1} \, \mathrm{d} t \mathrm{d} x$ where $w_{e^-_i} \coloneqq 0$ if
            $x_i < 1$.
        \item $e$ earns a utility of $\left(w_e - \sum_{i \in e} \int_{0}^{w_e}
            (k \ln(k))^{f_i(t) - 1} \, \mathrm{d} t\right) \mathrm{d} x$.
        \item The remaining
            \begin{align*}
                &\phantom{= } \left(\sum_{i \in e} \int_0^{w_e} (k \ln(k))^{f_i(t) - 1} \,
                \mathrm{d} t - \sum_{i \in e} \int_{w_{e^-_i}}^{w_e} (k \ln(k))^{f_i(t) - 1} \,
                \mathrm{d} t \right) \mathrm{d} x \\
                &= \sum_{i \in e} w_{e^-_i} \mathrm{d} x
            \end{align*}
            can be thought of as \enquote{paying off} the edges $e^-_i$ for the decrease in
            $y_{e^-_i}$.
    \end{enumerate}

    This means that as $e$ gets matched by some $\mathrm{d} x$, the total increase in the weight of
    the matching is $\left(w_e - \sum_{i \in e} w_{e^-_i}\right) \mathrm{d} x$ which is also the
    total increase in all revenues and utilities.
    Thus, if $r_i$ is the final revenue of $i \in V$ and $u_e$ is the final utility of $e \in E$,
    this implies that $\ALG = \sum_{i \in V} r_i + \sum_{e \in E} u_e$.
    Recall also from the proof of Theorem~\ref{thm:frac_lb}, that this, together with the following
    claim, establishes the result.

    \medskip
    \noindent
    \textbf{Claim:} For any $e \in E$, we have
    \[
        u_e + \sum_{i \in e} r_i \geq w_e \frac{1 - 1 / \ln(k)}{\ln(k) + \ln(\ln(k))}.
    \]
    \medskip

    \noindent
    \textbf{Proof of Claim:}
    Let $f_i$ be the step function defined in Algorithm~\ref{alg:weighted_wf} at the end of the
    algorithm.
    Then because prices (and utilities) are non-decreasing, we know that
    \[
        u_e \geq w_e - \sum_{i \in e} \int_0^{w_e} (k \ln(k))^{f_i(t) - 1} \, \mathrm{d} t
    \]
    On the other hand, the total revenue collected by each $i \in V$ satisfies
    \begin{align*}
        r_i &= \int_0^\infty \int_0^{f_i(t)} (k \ln(k))^{s - 1} \, \mathrm{d} s \,
        \mathrm{d} t \\
        &= \int_0^\infty \frac{\sum_{i \in e} (k \ln(k))^{f_i(t) - 1} -
        1 / \ln(k)}{\ln(k) + \ln(\ln(k))} \, \mathrm{d} t
    \end{align*}

    Now let $P(t) \coloneqq \sum_{i \in e} (k \ln(k))^{f_i(t) - 1}$, then this means
    that
    \begin{align*}
        u_e + \sum_{i \in e} r_i &= \max \left\{0, w_e - \int_0^{w_e} P(t) \, \mathrm{d} t \right\}
        + \int_0^\infty \frac{P(t) - 1 / \ln(k)}{\ln(k) + \ln(\ln(k))} \, \mathrm{d} t \\
        &\geq \max \left\{0, w_e - \int_0^{w_e} P(t) \, \mathrm{d} t \right\} + \int_0^{w_e}
        \frac{P(t) - 1 / \ln(k)}{\ln(k) + \ln(\ln(k))} \, \mathrm{d} t \\
        &= \max \left\{0, w_e - \int_0^{w_e} P(t) \, \mathrm{d} t \right\}
        + \frac{\int_0^{w_e} P(t)\, \mathrm{d}t - w_e / \ln(k)}{\ln(k) + \ln(\ln(k))}.
    \end{align*}
    Since this expression is minimized for $\int_0^{w_e} P(t) \, \mathrm{d} t = w_e$, the claim and
    thus the theorem follows.
\end{proof}

\section{Conclusion}

In this paper we have given a tight asymptotic bound for the fractional $k$-uniform hypergraph
matching problem and an almost tight bound for the integral variant.
This leaves room for some interesting directions for future research:

A major open problem is to beat the $\frac{1}{k}$ lower bound in the
integral setting. In fact, just recently Gamlath et al.\ \cite{gamlath2019online} showed that for
$k = 2$, no algorithm beats $\frac{1}{2}$, even if the underlying graph is bipartite.
However, their construction in fact shows this result for the fractional setting and where we know
how to beat $\frac{1}{k}$ for large $k$.
It thus remains open whether $\frac{1}{k} + \epsilon$ is achievable for \emph{any} $k$. In fact, for small $k$, one may also explicitly distinguish between edge arrival and vertex arrival models
as mentioned in Section~\ref{sec:model} or even the fully-online arrival model of Huang et al.\
\cite{huang2020fully}.
To the best of our knowledge, these models have only been explored for $k = 2$ and showing any
improvement over $\frac{1}{k}$ for $k > 2$ even under additional assumptions such as
$k$-partiteness would be rather interesting.

Finally, we remark that even in the fractional setting, exactly tight bounds are only known for $k
= 2$ and finding a tight non-asymptotic result remains open. 

\subsubsection{Acknowledgements} We would like to thank Vijay Vazirani for several helpful
comments during both the research and writing process.

%
%
%
\bibliographystyle{splncs04}
\bibliography{references}

\end{document}